\documentclass[12pt]{article}
\usepackage{amsmath}
\usepackage{setspace, lscape}
\RequirePackage{amsthm,amsmath,amsfonts,amssymb,mathtools}
\RequirePackage[authoryear]{natbib}
\RequirePackage[colorlinks,citecolor=blue,urlcolor=blue]{hyperref}
\RequirePackage{graphicx}
\newcommand{\blind}{1}
\usepackage{appendix}
\usepackage{booktabs}
\usepackage{enumitem}

\usepackage{bm}
\usepackage{float}
\usepackage{hhline}
\usepackage{rotating} 
\usepackage{multirow}
\usepackage{adjustbox}
\usepackage{algorithm}
\usepackage{algpseudocode}
\usepackage[normalem]{ulem}
\usepackage{xr}
\usepackage{subfig}

\newtheorem{theorem}{Theorem}
\newtheorem{lemma}{Lemma}
\newtheorem{proposition}{Proposition}
\newtheorem{assumption}{Assumption}
\newtheorem{corollary}{Corollary}
\newtheorem{remark}{Remark}
\newtheorem{definition}{Definition}

\DeclareMathOperator*{\argmin}{arg\,min}
\DeclareMathOperator{\KL}{KL}

\usepackage{tikz}
\usepackage{forest} 

\newcommand{\bx}{\bm{x}}

\usepackage{tabularx}
\usepackage[table]{xcolor}
\usepackage{colortbl}

\begin{document}

\def\spacingset#1{\renewcommand{\baselinestretch}%
{#1}\small\normalsize} \spacingset{1}


\if1\blind
{
  \title{\bf Structure Learning for Directed Trees with Zero-Inflated Compositional Nodes}
  
  \author{
        Shuangjie Zhang
        \thanks{Address for Correspondence:  105 E 24th St, Austin, TX, 78712. E-mail: shuangjie.zhang@austin.utexas.edu. }\\
        {\small Department of Statistics and Data Sciences, The University of Texas at Austin}
        \and
        Bani K. Mallick \\
        {\small Department of Statistics,  Texas A \& M University}
        \and
        Yang Ni \\
        {\small Department of Statistics and Data Sciences, The University of Texas at Austin}
}

  \maketitle
} \fi

\if0\blind
{
  \bigskip
  \bigskip
  \bigskip
  \begin{center}
    {\LARGE\bf Title}
\end{center}
  \medskip
} \fi

\bigskip
\begin{abstract}
Compositional data, which are vectors of proportions constrained to the probability simplex, arise frequently in modern scientific applications, including microbiome relative abundances across body sites and cell-type mixture weights derived from single-cell genomics. While regression methods for compositional data are well developed, no existing graphical model framework addresses the problem of learning conditional dependence structures among multiple compositional vectors. This paper introduces a novel framework for directed tree structure learning over compositional nodes. We employ the Kullback–Leibler divergence as the scoring function and model the conditional expectation of each child composition as a mixture of a baseline composition and a parent-driven component parameterized by a column-stochastic transition matrix. This formulation respects the simplex geometry, handles zero-inflated compositions gracefully, and, combined with a non-degeneracy condition on the transition matrix, ensures identifiability of edge directions from observational data. We prove consistency of structure recovery and derive finite-sample guarantees that characterize the required sample size in terms of the signal gap, node dimension, and penalty level. The efficacy of our approach is demonstrated through simulations and applications to multi-site microbiome data and single-cell data, yielding interpretable directed structures that align with known biological mechanisms.
\noindent
\end{abstract}

\noindent%
  \vfill

\newpage
\spacingset{1.75}

\section{Introduction}\label{sec:intro}

Zero-inflated compositional data, where each observation is a non-negative vector that sums to a constant, arise naturally across many scientific domains. A prominent example is microbiome data obtained via high-throughput sequencing \citep{jovel2016characterization}, for which the total sequencing depth varies arbitrarily across samples, rendering absolute counts uninformative (i.e., the biological signal is carried entirely by relative abundances \citep{gloor2017microbiome}), and many microbes have exact zero counts. Similarly, in single-cell genomics, cell-type compositions derived from clustering or classification procedures produce vectors of mixture weights that reside on the probability simplex. More broadly, any data pipeline that aggregates raw measurements into estimated histograms or mixture proportions yields compositional outputs. Because these data are constrained to lie on the simplex $\mathcal{S}^{d-1} = \{\bx \in \mathbb{R}^d : x_r \ge 0,\, \sum_{r} x_r = 1\}$, standard statistical techniques that treat them as unconstrained Euclidean vectors, such as Pearson correlation or ordinary least squares, are known to produce spurious results, most notably artificial negative correlations induced by the sum-to-one constraint \citep{aitchison1982statistical, friedman2012inferring}. To address this, a rich body of regression methodology has been developed for compositional data. In the scalar-on-composition setting, the log-contrast model \citep{aitchison1984log} and its high-dimensional extensions \citep{lin2014variable, shi2016regression,zhang2021bayesian} regress a scalar response on compositional covariates while respecting the simplex geometry. In the reverse direction, Dirichlet regression and related approaches model a compositional response as a function of scalar or categorical predictors \citep{hijazi2009modelling,koslovsky2020bayesian}. More recently, \citet{fiksel2022transformation} proposed a transformation-free linear regression for the composition-on-composition setting, using a Markov transition matrix to directly relate a compositional predictor to a compositional response. While these methods provide powerful tools for regression relationships, they do not address the problem of jointly modeling the conditional dependence structure among multiple compositions.

When multiple compositional vectors are observed simultaneously, for instance, microbial communities profiled across several body sites within the same individual, or cell-type mixture weights measured across immune cell lineages and subjects, understanding the conditional dependencies among these compositions becomes scientifically compelling. Our motivating applications include the Multi-Omic Microbiome Study-Pregnancy Initiative (MOMS-PI) \citep{fettweis2019vaginal}, which profiles microbial communities at five maternal body sites, and a single-cell eQTL dataset OneK1K \citep{yazar2022single}, in which immune cell subtypes yield compositional mixture weights across subjects. In both settings, the fundamental question is not merely whether two compositions are associated, but what directional dependence structure governs their relationships: does the vaginal microbiome drive the cervical microbiome, or vice versa? Do naive immune cells give rise to memory subtypes? Answering such questions requires a framework that can learn the conditional dependence relationships among compositional vectors.

Graphical models provide a natural probabilistic framework for representing conditional (in)dependence among random variables. In an undirected graphical model (also known as a Markov random field), the absence of an edge between two nodes encodes their conditional independence given all remaining variables. Directed acyclic graphical models (DAGs), also known as Bayesian networks, similarly encode conditional independence through the directed Markov property, but additionally allow for causal interpretation under the causal Markov assumption \citep{spirtes2000causation, koller2009probabilistic}. The field of DAG structure learning has seen extensive development, with approaches broadly categorized into constraint-based methods and score-based methods. Constraint-based methods, such as the PC algorithm \citep{spirtes2000causation} and its order-independent variant PC-Stable \citep{colombo2014order}, infer the graph structure by performing a sequence of conditional independence tests. Score-based methods formulate structure learning as an optimization problem solved by e.g., greedy equivalence search \citep{chickering2002optimal}, greedy hill-climbing \citep{bouckaert1994properties}, and the continuous optimization approach \citep{zheng2018dags}; each evaluates candidate graphs using a goodness-of-fit criterion and searches the graph space for an optimum. A fundamental challenge in structure learning is identifiability: from observational data alone, the true DAG can generally only be recovered up to a Markov equivalence class (MEC) \citep{spirtes2000causation}. Distinguishing among graphs within an MEC typically requires additional assumptions, such as non-Gaussianity or nonlinearity \citep{shimizu2006linear, hoyer2008nonlinear, peters2014causal}.

Existing graphical models, both undirected and directed, have been developed primarily for scalar-valued nodes, with classical methods focusing on Gaussian or categorical variables \citep{cooper1992bayesian, koller2009probabilistic}. Recent extensions have broadened the scope to functional data, where each node represents a random function such as a time series or a curve \citep{lee2022functional, yang2022functional, zhou2023functional,roy2023directed}. Graphical models have also been applied to compositional data, but with a fundamentally different objective: existing methods model the conditional dependence among individual microbial taxa within a single microbial community, treating each taxon's abundance as a scalar node \citep{ma2021joint,chung2022phylogenetically}. In contrast, our goal is to learn the dependence structure among multiple compositional vectors, for example, entire microbial communities across body sites, where each node is itself a composition on the simplex. To the best of our knowledge, there is no existing undirected or directed graphical models where each node represents a compositional vector on the probability simplex. This paper is the first attempt in this direction.

Within the class of DAGs, directed trees have been extensively studied and enjoy a distinctive computational advantage. The classical algorithms of \citet{chow1968approximating} and \citet{edmonds1967optimum} guarantee recovery of the globally optimal tree structure by efficiently solving a maximum weight spanning arborescence problem. No analogous polynomial-time algorithm with global optimality guarantees exists for general DAGs, where the search space grows super-exponentially with the number of nodes and the optimization is NP-hard \citep{chickering1996learning}. The computational tractability of directed trees, combined with the interpretability of tree topologies in domains such as microbial phylogenetics and cell lineage, motivates our focus on directed tree models in this paper.

Learning directed trees over compositional nodes introduces three main challenges. First, the simplex geometry must be respected: because compositions are constrained vectors, the loss/score function must operate natively on the simplex rather than treating the data as unconstrained Euclidean vectors. Second, real-world compositional data are often severely zero-inflated. For example, the MOMS-PI microbiome data exhibit zero rates ranging from 39\% to 67\% even after filtering for prevalence. Therefore, the methodology should be able to handle exact zeros. Third, identifiability of edge directions is nontrivial: directed trees sharing the same skeleton belong to the same Markov equivalence class and encode identical conditional independence structures. Hence, distinguishing between alternative orientations requires exploiting specific properties of the loss/score function.

In this paper, we propose a novel framework for learning directed tree structures over zero-inflated compositional nodes. We define a score function based on the Kullback-Leibler (KL) divergence, which provides a natural measure of discrepancy for compositional data while inherently respecting the geometry of the simplex and gracefully handling the zero inflation. For each candidate parent-child pair, the conditional expectation of the child composition is modeled as a mixture of a baseline composition and a parent-driven component parameterized by a column-stochastic transition matrix. This formulation ensures that predictions remain on the simplex. We establish identifiability of the true directed tree under mild non-degeneracy conditions and prove consistency of the structure recovery procedure, supplemented by finite-sample guarantees that quantify the required sample size in terms of the signal gap, node dimension, and penalty level. For parameter estimation, we exploit the connection between KL minimization and multinomial quasi-likelihood to derive an efficient Expectation-Maximization (EM) algorithm, and we adopt the Chu--Liu/Edmonds algorithm, augmented with a virtual root to accommodate directed forests (multiple disconnected directed trees), for the tree search. Cross-validation is used to select the regularization parameter that controls the density of the learned tree/forest. We demonstrate the efficacy of the proposed method through simulation studies and apply it to both multi-domain microbiome data and single-cell data.

The remainder of the paper is organized as follows. Section~\ref{sec:model} introduces the KL-based loss function, establishes the graph identifiability under the proposed loss function, and establishes the estimation theories, including consistency of structure recovery and finite-sample guarantees. Section~\ref{sec:algorithm} describes the EM algorithm for parameter estimation and the overall structure learning procedure with cross-validation. Section~\ref{sec:sim} presents simulation studies, and Section~\ref{sec:realdata} applies the method to the MOMS-PI microbiome dataset and a single-cell eQTL dataset. Section~\ref{sec:con} concludes with a discussion of future directions.

\section{Method}\label{sec:model}
\subsection{Proposed Score Function}\label{sec:psf}
Let $T = (V, E)$ be a directed tree where $V = \{1, \dots, p\}$ is the set of nodes and $E \subseteq V \times V$ is the set of directed edges.  For a directed edge $(j,k)\in E$, we call node $j$ a parent of node $k$, and node $k$ a child of node $j$. A directed tree is a DAG for which each node has at most one parent\footnote{Note that we do not distinguish between a tree (one root) and a forest (multiple roots)}. We denote the parent of node $j$ in $T$ as $\pi_j$. For a directed tree, $\pi_j$ is either a singleton or an empty set. Unlike standard graphical models, where each node $v \in V$ represents a scalar random variable, in this paper, each node represents a random composition vector with $d_j$ parts, denoted by $\mathbf{X}^{(j)} \in \mathcal{S}^{d_j-1}$ where $\mathcal{S}^{d_j-1}= \left\{ \mathbf{X} \in \mathbb{R}^{d_j} \mid X_r \ge 0 \text{ for all } r, \sum_{r=1}^{d_j} X_r = 1 \right\}$ is a $d_j-1$ dimensional simplex.

We define the following population-level loss function for an ordered pair of nodes $(j,k)$ based on Kullback-Leibler (KL) divergence $D_{KL}(\cdot\parallel\cdot)$:
\begin{equation}\label{eq:pnloss}
\mathcal{L}_{jk}\left(\mathbf{X}^{(j)}\mid \mathbf{X}^{(k)}\right) = D_{KL} \left( \mathbf{X}^{(j)} \parallel \widehat{\mathbf{X}}^{(j)} \right) = - \sum_{r=1}^{d_j} X_r^{(j)} \log \left( \frac{\widehat{X}_{r}^{(j)}}{X_{r}^{(j)}} \right),
\end{equation}
which measures the discrepancy between $\mathbf{X}^{(j)}$ and its conditional expectation $\widehat{\mathbf{X}}^{(j)}=\mathbb{E}[\mathbf{X}^{(j)}|\mathbf{X}^{(k)}]$. Unlike the Dirichlet distribution, which has been commonly used for compositional data, this loss \citep{fiksel2022transformation} is particularly appropriate for zero-inflated compositional data as it handles zeros gracefully; since $\lim_{x \to 0^+} x \log x = 0$, the loss does not diverge as $X_r^{(j)}\to 0^+$. 

To respect the compositional nature of $\mathbf{X}^{(j)}$, we specify the conditional expectation as,
\begin{equation} \label{eq:condex}
\widehat{\mathbf{X}}^{(j)} = \omega^{(j)}_{0} \boldsymbol{\eta}_j + \omega^{(j)}_{1} \left( \mathbf{M}_{jk} \mathbf{X}^{(k)} \right),
\end{equation}
subject to the constraints  $\omega^{(j)}_{0}\ge 0$, $\omega^{(j)}_{1} > 0$, and $\omega^{(j)}_{0} + \omega^{(j)}_{1} = 1$.
Note that we restrict $\omega^{(j)}_{1}$ to be strictly positive in anticipating the later use of \eqref{eq:pnloss}-\eqref{eq:condex} in the definition of the score function of a directed tree to avoid ambiguity/contradiction that an edge is present but the corresponding weight is zero. 
In \eqref{eq:condex}, $\boldsymbol{\eta}_j \in \mathcal{S}^{d_j-1}$ is the \textit{baseline} composition for node $j$, i.e., the composition of node $j$ in the absence of the influence from node $k$. The matrix $\mathbf{M}_{jk} \in \mathbb{R}^{d_j \times d_k}$ is a column-stochastic \textit{transition matrix} (i.e.,  each column of $\mathbf{M}_{jk} $ lies in $\mathcal{S}^{d_j-1}$), which maps the composition of node $k$ to that of node $j$. The scalars $\omega^{(j)}_{0}$ and $\omega^{(j)}_{1}$ quantify the mixing proportions, representing the relative contribution of the baseline composition for node $j$ and the composition of node $k$. This formulation ensures that $\widehat{\mathbf{X}}^{(j)}$ always resides on the simplex $\mathcal{S}^{d_j-1}$. 
With a slight abuse of notation, \eqref{eq:pnloss} also denotes the loss of node $j$ conditional on the empty set by taking $\widehat{\mathbf{X}}^{(j)} = \boldsymbol{\eta}_j$ to be the marginal expectation of $\mathbf{X}^{(j)}$.

We now introduce a score-based tree structure learning method by defining the population-level loss function for tree $T$ and compositions $\mathbf{X}=\{\mathbf{X}^{(j)}\}_{j=1}^p$,
\begin{align}\label{eq:ptloss}
    \mathcal{L}(\mathbf{X};T,\Theta) = \sum_{j=1}^p\mathcal{L}_{j\pi_j}(\mathbf{X}^{(j)}\mid\mathbf{X}^{(\pi_j)},\Theta_{j\pi_j}),
\end{align}
where $\Theta=\{\Theta_{j\pi_j}\}_{j=1}^p$ and we make explicit that the loss in \eqref{eq:pnloss} is parameterized by $\Theta_{j\pi_j}=\{\omega_0^{(j)},\omega_1^{(j)},\boldsymbol{\eta}_j,\mathbf{M}_{j\pi_j}\}$ if $\pi_j\neq \emptyset$, or $\Theta_{j\pi_j}=\{\boldsymbol{\eta}_j\}$ if $\pi_j=\emptyset$.  The tree loss \eqref{eq:ptloss} measures the overall goodness-of-fit of $\mathbf{X}$ given a tree $T$ and its associated parameter set $\Theta$. Taking the expectation of \eqref{eq:ptloss} with respect to $\mathbf{X}$ gives the risk function,
\begin{align*}
    \mathcal{R}(\mathbf{X};T,\Theta) &= \mathbb{E}\left[\sum_{j=1}^p\mathcal{L}_{j\pi_j}(\mathbf{X}^{(j)}\mid\mathbf{X}^{(\pi_j)},\Theta_{j\pi_j})\right],\\
    &=\sum_{j=1}^p\mathcal{R}_{j\pi_j}(\mathbf{X}^{(j)}\mid\mathbf{X}^{(\pi_j)},\Theta_{j\pi_j}),
\end{align*}
where $\mathcal{R}_{j\pi_j}(\mathbf{X}^{(j)}\mid\mathbf{X}^{(\pi_j)},\Theta_{j\pi_j})=\mathbb{E}[\mathcal{L}_{j\pi_j}(\mathbf{X}^{(j)}\mid\mathbf{X}^{(\pi_j)},\Theta_{j\pi_j})]$.

Given a sample of $\mathbf{X}$, denoted by $\mathbf{x}_1,\dots,\mathbf{x}_n$, the sample-level risk or the averaged loss is then given by,
\begin{align}\label{eq:stloss}
    \mathcal{R}(\mathbf{x}_1,\dots,\mathbf{x}_n;T,\Theta) = \frac{1}{n}\sum_{i=1}^n\sum_{j=1}^p\mathcal{L}_{j\pi_j}(\mathbf{x}_i^{(j)}\mid\mathbf{x}_i^{(\pi_j)},\Theta_{j\pi_j}).
\end{align}
The proposed score function of a tree $T$ is then the risk evaluated at the optimal $\Theta$ plus a complexity penalty, 
 \begin{equation*}
    S(T;\alpha)= \min_{\Theta} \mathcal{R}(\mathbf{x}_1,\dots,\mathbf{x}_n;T,\Theta) + \alpha \cdot |E|,
\end{equation*}
where $|E|$ is the number of edges in $T$ and $\alpha\geq0$ is a regularization parameter that penalizes complex trees (smaller $\alpha$ leads to a denser tree). Our goal is then to search for $T$ that minimizes $S(T;\alpha)$, that is, $\widehat{T} = \arg\min_{T \in \mathcal{T}}S(T;\alpha)$ where $\mathcal{T}$ is the set of all directed trees on $V = \{1,\ldots,p\}$. This procedure gives rise to two separate optimization problems involving an inner minimization for parameter $\Theta$ given a tree $T$ and an outer minimization for tree search, which will be detailed in Section~\ref{sec:algorithm}.

\subsection{Identifiability} 

Directed trees, more generally DAGs, in the same MEC have exactly the same Markov properties (i.e., conditional independence assertions) and hence are generally not distinguishable from each other. For example, Gaussian and multinomial directed tree models are non-identifiable.
The proposed score, however, goes beyond just conditional independence and allows the tree structure to be uniquely identifiable.

\begin{proposition}[Tree Identifiability]\label{prop:ident_weak}
Under the proposed score function, the directed tree structure $T$ is identifiable (i.e., no two distinct trees lead to the same score), provided that the transition matrices $\mathbf{M}_{jk}$'s are non-degenerate, meaning that not all columns of $\mathbf{M}_{jk}$ are identical.
\end{proposition}

We prove Proposition~\ref{prop:ident_weak} by using the property that Markov equivalent directed trees must share the same skeleton. While different orientations of a fixed skeleton are statistically indistinguishable in terms of conditional independence, the proposed score breaks this symmetry unless the parent effect degenerates (see Remark \ref{re:sharp}). In particular, for any non-root nodes, the loss distinguishes the true parent from alternative candidates because the parent-driven component induces a dependence that cannot be replicated by reversing the edge. The detailed argument is provided in the Web Appendix A.

\begin{remark}[Sharpness of the non-degeneracy condition]\label{re:sharp}
The non-degeneracy condition is tight. If $\mathbf{M}_{jk}$ has all identical columns $(\mathbf{M}_{jk} = \mathbf{m}\mathbf{1}^\top)$, then $\mathbf{M}_{jk} \mathbf{X}^{(k)} = \mathbf{m}$ for every $\mathbf{X}^{(k)} \in \mathcal{S}^{d_k-1}$, meaning the parent-driven component is constant. In this case, the conditional expectation of $\mathbf{X}^{(j)}$ can be absorbed into a single baseline $\widehat{\mathbf{X}}^{(j)} = \omega_0 \boldsymbol{\eta}_j + \omega_1 \mathbf{m}:=\widetilde{\boldsymbol{\eta}}_j \in \mathcal{S}^{d_j-1}$, and the edge $k \to j$ is observationally equivalent to $j$ being a root. Reversing the edge would give the same loss, and hence identifiability fails.

\end{remark}

\subsection{Estimation Theories}\label{sec:proof}

In this section, we establish theoretical guarantees for recovering the true directed tree. We first show that, under a suitable separation condition, the true tree uniquely minimizes the population criterion. Then, we prove that the empirical objective converges uniformly to this population target, yielding consistent recovery of the underlying structure.

\subsubsection{Consistency of Structure Recovery}

Let $T^* = (V, E^*) \in \mathcal{T}$ denote the true data-generating structure, and let $P^*$ denote the true data-generating distribution of $\mathbf{X}$, which is Markov with respect to $T^*$, leading to the following factorization:
\begin{align*}
    P^*(\mathbf{X})=\prod_{j\in V}P_{j\pi_j^*}^*(\mathbf{X}_j|\mathbf{X}_{\pi_j^*}),
\end{align*}
where $\pi_j^*$ denotes the true parent of $j$ in $T^*$ (empty if $j$ is a root) and $P_{j\pi_j^*}^*(\cdot|\cdot)$ is the conditional distribution under $P^*$. We leave $P_{j\pi_j^*}^*(\cdot|\cdot)$ unspecified except for its first moment: for each edge $(k,j) \in E^*$, we assume the true conditional expectation of $\mathbf{X}^{(j)}$ is given by \eqref{eq:condex} with true parameters $\Theta^*_{j\pi_j} = \{\omega_0^{(j)*}, \omega_1^{(j)*}, \boldsymbol{\eta}_j^*, \mathbf{M}_{j\pi_j}^*\}$ if $\pi_j^*\neq \emptyset$  and $\widehat{\mathbf{X}}_i^{(j)} = \boldsymbol{\eta}_j^*$ otherwise. To analyze structure recovery, we study the population-level risk function $\mathcal{R}_{jk}(\Theta_{jk}):=\mathcal{R}_{jk}(\mathbf{X}^{(j)}\mid\mathbf{X}^{(k)},\Theta_{jk})=\mathbb{E}[\mathcal{L}_{jk}(\mathbf{X}^{(j)}\mid\mathbf{X}^{(k)},\Theta_{jk})]$ as defined in Section \ref{sec:psf}.
Let $\mathcal{R}_{jk}^*$ denote the \emph{optimized risk},
\begin{equation*}
\mathcal{R}_{jk}^* = \inf_{\Theta_{jk} \in \Xi_{jk}} \mathcal{R}_{jk}(\Theta_{jk}),
\end{equation*}
where $\Xi_{jk}$ is the (compact) parameter space for $\Theta_{jk}$. When node $j$ is a root, the optimized risk reduces to
\begin{equation*}
\mathcal{R}_{j}^{\mathrm{root}} = \inf_{\boldsymbol{\eta}_j \in \mathcal{S}^{d_j-1}} \mathbb{E}\left[ D_{\KL}\!\left(\mathbf{X}^{(j)} \,\|\, \boldsymbol{\eta}_j\right) \right].
\end{equation*}

The corresponding \emph{sample versions} are
\begin{equation*}
\widehat{\mathcal{R}}_{jk}(\Theta_{jk}) = \frac{1}{n} \sum_{i=1}^n D_{\KL}\!\left(\mathbf{x}_i^{(j)} \,\|\, \widehat{\mathbf{x}}_i^{(j)}\right), \qquad \widehat{\mathcal{R}}_{jk}^* = \inf_{\Theta_{jk} \in \Xi_{jk}} \widehat{\mathcal{R}}_{jk}(\Theta_{jk}),
\end{equation*}
and $\widehat{\mathcal{R}}_j^{\mathrm{root}}$ is defined analogously. Optimizing the score involves the comparison of the risks for a node with or without a parent. This motivates the following notion of edge signal.

\begin{definition}[Edge Signal Strength]\label{def:signal}
The \emph{signal strength} of a candidate edge $k \to j$ is the population-level KL reduction from including the edge versus treating $j$ as a root:
\begin{equation*}
\delta_{jk} = \mathcal{R}_j^{\mathrm{root}} - \mathcal{R}_{jk}^*.
\end{equation*}
By construction, $\delta_{jk} \geq 0$, with equality if and only if the parent node $k$ provides no information for node $j$.
\end{definition}

For a true edge $(k,j)\in E^*$, the parent-driven
component introduces genuine predictive information, and hence the KL reduction
is strictly positive, $\delta_{jk}>0$. In contrast, for a non-edge $(k',j)\notin
E^*$, the signal $\delta_{jk'}$ reflects only indirect association.
It is zero when $k'$ and $j$ are marginally independent, and may be positive
when $k'$ is an ancestor or shares common ancestry with $j$. Such indirect
effects are strictly weaker under our model and are discussed in
Remark~\ref{rem:grandparent} below.

The population-level penalized score for a candidate tree $T = (V, E)$ is
\begin{equation*}
S(T;\alpha) = \sum_{j \in V} \mathcal{R}_j^{\mathrm{root}} - \sum_{(k,j) \in E} \delta_{jk} + \alpha \cdot |E|
= \underbrace{\sum_{j \in V} \mathcal{R}_j^{\mathrm{root}}}_{\text{constant}} + \sum_{(k,j) \in E} \left(\alpha - \delta_{jk}\right),
\end{equation*}
where we have used $\mathcal{R}_{jk}^* = \mathcal{R}_j^{\mathrm{root}} - \delta_{jk}$ for nodes with a parent and $\mathcal{R}_j^{\mathrm{root}}$ for root nodes. This representation makes the role of $\alpha$ clear: an edge $(k,j)$ reduces the score if and only if $\delta_{jk} > \alpha$.

The sample-level score is
\begin{equation}\label{eq:sample_score}
\widehat{S}_n(T;\alpha) = \sum_{(k,j) \in E} \widehat{\mathcal{R}}_{jk}^* + \sum_{j:\,\pi_j=\varnothing} \widehat{\mathcal{R}}_j^{\mathrm{root}} + \alpha \cdot |E|.
\end{equation}
The estimated tree is $\widehat{T}_n = \argmin_{T \in \mathcal{T}} \widehat{S}_n(T;\alpha)$. This decomposition shows that the global optimization separates into independent
edge contributions. An edge is selected precisely when its signal exceeds the
penalty level, making the procedure analogous to thresholding in sparse model
selection. To guarantee the consistency of structure recovery, we require the following assumptions. 

\begin{assumption}[Compactness and Boundedness]\label{ass:compact}
For any ordered pair $(k,j)$, the parameter space $\Xi_{jk}$ is compact. The observations satisfy $x_{i,r}^{(j)} \in [0,1]$ and $\sum_r x_{i,r}^{(j)} = 1$ for all $i, j, r$. There exists $\epsilon_0 > 0$ such that for all $\Theta_{jk} \in \Xi_{jk}$ and all $\mathbf{x}_i^{(k)} \in \mathcal{S}^{d_k-1}$, we have $\widehat{x}_{i,r}^{(j)} \geq \epsilon_0$ for all $r = 1, \ldots, d_j$.
\end{assumption}

The lower bound $\widehat{x}_r^{(j)} \geq \epsilon_0$ is ensured whenever the parameter space enforces $\omega_0^{(j)} \geq \omega_{\min} > 0$ and $\eta_{j,r} \geq \eta_{\min} > 0$ for all $r$, giving $\widehat{x}_{i,r}^{(j)} \geq \omega_{\min} \cdot \eta_{\min}$. This is a mild restriction that excludes degenerate baseline distributions.

\begin{assumption}[Signal Separation]\label{ass:separation}
There exists a threshold $\alpha^* > 0$ such that the true and false edge signals are separated:
\begin{equation*}
\delta_{\min}^+ := \min_{(k,j) \in E^*} \delta_{jk} > \alpha^* > \max_{(k',j):\, k' = \argmin_{k \neq \pi_j^*} \mathcal{R}_{jk}^* } \delta_{jk'} =: \delta_{\max}^-.
\end{equation*}
\end{assumption}

\begin{lemma}[Population optimality under signal separation]\label{lem:pop_opt}
Suppose Assumption~\ref{ass:separation} holds, i.e., there exists
$\alpha^* \in (\delta_{\max}^-,\delta_{\min}^+)$ such that
\[
\min_{(k,j)\in E^*}\delta_{jk}>\alpha^*>\max_{(k',j):\, k' = \argmin_{k \neq \pi_j^*} \mathcal{R}_{jk}^* } \delta_{jk'}.
\]
Then the true tree $T^*$ is the unique minimizer of the population penalized criterion:
\[
T^*=\arg\min_{T\in\mathcal T} S(T;\alpha^*).
\]
Equivalently, for any $T\neq T^*$, $S(T;\alpha^*)>S(T^*;\alpha^*)$.
\end{lemma}

\textit{Proof.}
Under Assumption~\ref{ass:separation} with $\alpha=\alpha^*$, every true edge satisfies $\delta_{jk}>\alpha^*$ and every false edge satisfies $\delta_{jk'}<\alpha^*$. Hence, for each node $j$, the parent choice in $T^*$ minimizes the node-wise contribution to $S(T;\alpha^*)$. Choosing a false parent (or incorrectly setting $j$ as a root) strictly increases the score, and dropping a true parent also increases the score. Since $S(T;\alpha^*)$ is a sum of these node-wise terms, any $T\neq T^*$ has strictly larger score. Therefore $T^*$ is the unique minimizer.
\hfill$\square$

Assumption~\ref{ass:separation} is thus the key population-level identifiability condition for score minimization: true edges are rewarded and false edges are penalized at $\alpha^*$. This is analogous to a ``beta-min'' condition in sparse model selection.

\begin{remark}[Indirect Associations and Ancestor Edges]\label{rem:grandparent}
A subtlety in Assumption~\ref{ass:separation} concerns indirect associations: if $k'$ is an ancestor of $j$ (e.g., the true structure contains a path $k' \to k \to j$), then $\mathbf{X}^{(k')}$ carries information about $\mathbf{X}^{(j)}$ through the intermediate node $k$, so the signal $\delta_{jk'}$ may be non-negligible. 

The true conditional expectations are
\begin{align*}
\mathbb{E}[\mathbf{X}^{(k)} | \mathbf{X}^{(k')}] &= \omega_0^{(k)} \boldsymbol{\eta}_k + \omega_1^{(k)} \mathbf{M}_{kk'}\, \mathbf{X}^{(k')}, \\
\mathbb{E}[\mathbf{X}^{(j)} | \mathbf{X}^{(k)}] &= \omega_0^{(j)} \boldsymbol{\eta}_j + \omega_1^{(j)} \mathbf{M}_{jk}\, \mathbf{X}^{(k)}. 
\end{align*}
By the Markov property, conditioning on $\mathbf{X}^{(k')}$ alone yields
\begin{equation*}
\mathbb{E}[\mathbf{X}^{(j)} | \mathbf{X}^{(k')}] = \omega_0^{(j)} \boldsymbol{\eta}_j + \omega_1^{(j)} \mathbf{M}_{jk}\left(\omega_0^{(k)} \boldsymbol{\eta}_k + \omega_1^{(k)} \mathbf{M}_{kk'}\, \mathbf{X}^{(k')}\right).
\end{equation*}
This shows that the grandparent edge $k' \to j$ can capture the conditional mean of $\mathbf{X}^{(j)}$ given $\mathbf{X}^{(k')}$, with an effective transition matrix $\widetilde{\mathbf{M}}_{jk'} = \mathbf{M}_{jk} \mathbf{M}_{kk'}$ and effective mixing weight $\widetilde{\omega}_1 = \omega_1^{(j)} \omega_1^{(k)}$. However, the KL loss depends not only on the conditional mean but also on the variability. By the data processing inequality for mutual information applied to the Markov chain,
\begin{equation}\label{eq:dpi}
I(\mathbf{X}^{(j)};\, \mathbf{X}^{(k)}) \geq I(\mathbf{X}^{(j)};\, \mathbf{X}^{(k')}),
\end{equation}
with equality if and only if $\mathbf{X}^{(k)}$ is a deterministic function of $\mathbf{X}^{(k')}$. Since the KL loss reduction $\delta_{jk}$ measures the predictive information that node $k$ provides about $j$ (relative to the baseline), inequality~\eqref{eq:dpi} implies $\delta_{jk} \geq \delta_{jk'}$
with strict inequality whenever $\mathbf{X}^{(k)}$ has genuine stochastic variation conditional on $\mathbf{X}^{(k')}$. 
\end{remark}

\begin{assumption}[IID Sampling]\label{ass:regularity}
The observations $(\mathbf{x}_1, \ldots, \mathbf{x}_n)$ are i.i.d.\ draws from the joint distribution $P^*$.
\end{assumption}



With these assumptions and the oracle regularization parameter, the tree structure can be consistently estimated. 
\begin{theorem}[Consistency of Structure Recovery]\label{thm:consistency}
Under Assumptions~\ref{ass:compact}--\ref{ass:regularity} with $\alpha=\alpha^*$, the estimated directed tree $\widehat{T}_n$ recovers the true structure $T^*$ with probability tending to one:
\begin{equation*}
\Pr\!\left(\widehat{T}_n = T^*\right) \to 1 \quad \text{as } n \to \infty.
\end{equation*}
\end{theorem}

The proof proceeds in two steps. First, we establish uniform convergence of the optimized sample risks to their population counterparts: under Assumptions \ref{ass:compact} and \ref{ass:regularity}, for every ordered pair $(j,k)$, the optimized sample risk $\widehat{\mathcal{R}}^*_{jk}$ converges in probability to $\mathcal{R}^*_{jk}$, and this convergence holds uniformly over all pairs. The argument combines the compactness of the parameter space $\Xi_{jk}$ (Assumption \ref{ass:compact}) with the boundedness of the KL loss, which is controlled by $\log(1/\epsilon_0)$, to verify a uniform law of large numbers via an $\epsilon$-net argument over $\Xi_{jk}$. The same reasoning applies to $\widehat{\mathcal{R}}^{\text{root}}_j$. Second, the uniform convergence of the risks implies uniform convergence of the estimated edge signals $\widehat{\delta}_{jk} = \widehat{\mathcal{R}}^{\text{root}}_j - \widehat{\mathcal{R}}^*_{jk}$ to their population values $\delta_{jk}$. Combined with Assumption \ref{ass:separation} (signal separation) and $\alpha=\alpha^*$, this ensures that for sufficiently large $n$, every true edge satisfies $\hat{\delta}_{jk} > \alpha$ and every false edge satisfies $\hat{\delta}_{jk'} < \alpha$, so the minimizer of the sample-level penalized score $\widehat{S}_n(T;\alpha)$ coincides with $T^*$. The complete proof is provided in the Web Appendix B.

Note that our asymptotic theory assumes knowledge of the oracle regularization parameter $\alpha=\alpha^*$. In practice, we select it via cross-validation. Establishing consistency under selected $\alpha$ is beyond the scope of this paper.

\subsubsection{Finite-Sample Guarantee}

We further extend the asymptotic result established above into quantitative bounds that characterize the performance of the estimator at finite sample sizes in the following Corollary~\ref{cor:finite}.

\begin{corollary}[Finite-Sample Recovery Guarantee]\label{cor:finite}
Under Assumptions~\ref{ass:compact}--\ref{ass:regularity} with $\alpha = \alpha^*$, 
\begin{equation}\label{eq:final_bound}
\Pr\!\left(\widehat{T}_n \neq T^*\right) \leq 4(p^2+p) \left(\frac{24RL}{\gamma}\right)^{D_{\max}} \exp\!\left(-\frac{n\gamma^2}{32(\log(1/\epsilon_0))^2}\right),
\end{equation}
where $\gamma = \min(\delta_{\min}^+ - \alpha^*,\; \alpha^* - \delta_{\max}^-)$ is the signal gap, $L = d_{\max}\sqrt{D_{\max}}/\epsilon_0$, $d_{\max}=\max_jd_j$, $D_{\max} = \max_{j,k} (1 + d_j + d_jd_k)$, and $R$ is the maximum diameter of the parameter space $\Xi_{jk}$ for all $j,k$.
\end{corollary}

Corollary~\ref{cor:finite} quantifies how large the sample size is needed to reliably distinguish true and false edges. The bound shows that recovery becomes exponentially more likely as $n$ increases, with the required sample size scaling as $\gamma^{-2}$, where
$\gamma$ is the signal gap between correct and incorrect parents. This dependence mirrors classical results in high-dimensional model selection: stronger separation between competing models leads to faster structure recovery, while higher-dimensional nodes increase the complexity through $D_{\max}$.

\begin{remark}[Sample Complexity]
For cleaner exposition, the bound \eqref{eq:final_bound} can be written as
\begin{equation*}
\Pr\!\left(\widehat{T}_n \neq T^*\right) \leq \exp\!\left(-\frac{n\gamma^2}{32(\log(1/\epsilon_0))^2} + D_{\max}\log\!\left(\frac{24RL}{\gamma}\right) + \log(4(p^2+p))\right).
\end{equation*}
This is less than $\delta$ when
\begin{equation*}
n \geq \frac{32(\log(1/\epsilon_0))^2}{\gamma^2}\left(D_{\max}\log\!\left(\frac{24RL}{\gamma}\right) + \log\!\left(\frac{4(p^2+p)}{\delta}\right)\right).
\end{equation*}
The sample complexity is thus $O\!\left(\gamma^{-2} \log^2(1/\epsilon_0)\left(D_{\max}\log(D_{\max}/\gamma) + \log p\right)\right)$, where the parameter dimension $D_{\max} \approx d_{\max}^2$ now appears explicitly through the covering number. The dominant term for large node dimensions is $D_{\max} \log(D_{\max}/\gamma)$, reflecting the cost of uniformizing over the parameter space.
\end{remark}


\section{Structure Learning Algorithm}\label{sec:algorithm}

Optimizing the proposed score function involves an inner minimization over the parameters for each tree and an outer minimization over the space of trees. The additivity of the tree loss~(\ref{eq:stloss}) across nodes/edges decouples these two stages: the parameters $\Theta_{jk}$ for each candidate edge $k \to j$ can be estimated independently, after which the globally optimal tree is recovered via a single call to the Chu--Liu/Edmonds algorithm. We describe each stage below.

\subsection{Parameter Estimation via EM}\label{sec:em}

For a fixed candidate edge $k \to j$, we must estimate the parameters $\Theta_{jk} = \{\omega^{(j)}_{0}, \omega^{(j)}_{1}, \boldsymbol{\eta}_j, \mathbf{M}_{jk}\}$ by minimizing the sample-level KL loss $\widehat{\mathcal{R}}_{jk}(\Theta_{jk})$. This minimization admits a natural interpretation as maximum likelihood estimation for a two-component mixture model. Specifically, note that
\[
\sum_{i=1}^n \sum_{r=1}^{d_j} x_{i,r}^{(j)} \log \widehat{x}_{i,r}^{(j)}
\]
is the log quasi-likelihood of a multinomial distribution in which the observed compositions $\mathbf{x}_i^{(j)}$ play the role of empirical frequency vectors and the predicted compositions $\widehat{\mathbf{x}}_i^{(j)}$ play the role of cell probabilities. Since the predicted composition~(\ref{eq:condex}) is a two-component mixture of a baseline $\boldsymbol{\eta}_j$ and a parent-driven component $\mathbf{M}_{jk}\mathbf{x}_i^{(k)}$, the optimization problem has the same structure as fitting a mixture model via maximum likelihood. This connection, analogous to the coarsened likelihood perspective of \cite{miller2019robust}, justifies the use of the Expectation--Maximization (EM) algorithm \citep{fiksel2022generalized}.

In the E-step, we compute, for each sample $i$ and part $r$, the responsibility of the baseline component, $\gamma_{i,r,0} = \omega_0^{(j)} \eta_{j,r} / \widehat{x}_{i,r}^{(j)}$, and the responsibility of the parent-driven component, $\gamma_{i,r,1} = \omega_1^{(j)} (\mathbf{M}_{jk}\mathbf{x}_i^{(k)})_r / \widehat{x}_{i,r}^{(j)}$. Within the parent-driven component, we further allocate across the parent's parts via $\pi_{i,r,j,c} = x_{i,c}^{(k)} M_{jk,cr} / (\mathbf{M}_{jk}\mathbf{x}_i^{(k)})_r$. In the M-step, all parameters admit closed-form updates: the mixing weights, the baseline composition $\boldsymbol{\eta}_j$, and the transition matrix $\mathbf{M}_{jk}$ are each updated as weighted averages dictated by the responsibilities. The complete procedure is summarized in Algorithm~\ref{alg:em}.


\begin{algorithm}
\caption{EM Algorithm for Pairwise Parameter Estimation}\label{alg:em}
\begin{algorithmic}[1]
\Require Child data $\mathbf{X}^{(j)}$, Parent data $\mathbf{X}^{(k)}$, Tolerance $\epsilon$
\Ensure Estimated parameters $\widehat{\Theta}_{jk} = \{\widehat{\omega}^{(j)}_{0}, \widehat{\omega}^{(j)}_{1}, \widehat{\boldsymbol{\eta}}_j, \widehat{\mathbf{M}}_{jk}\}$
\State Initialize parameters $\Theta^{(0)}$.
\Repeat
    \State \textbf{E-Step:} Compute latent responsibilities for each sample $i$ and part $r$:
    \State $\gamma_{i,r,0} \gets \frac{\omega^{(j)}_{0} \eta_{j,r}}{\widehat{x}_{i,r}^{(j)}}$ \Comment{Baseline responsibility}
    \State $\gamma_{i,r,1} \gets \frac{\omega^{(j)}_{1} ( \mathbf{M}_{jk}\mathbf{x}_i^{(k)})_r}{\widehat{x}_{i,r}^{(j)}}$ \Comment{Parent responsibility}
    \State $\pi_{i,r,j,c} \gets \frac{x_{i,c}^{(k)} M_{jk,cr}}{(\mathbf{M}_{jk}\mathbf{x}_i^{(k)} )_r}$ \Comment{Part-level allocation}
    
    \State \textbf{M-Step:} Update parameters using closed-form solutions:
    \State $\omega^{(j)}_{0} \gets \frac{\sum_{i=1}^n \sum_{r=1}^{d_j} x_{i,r}^{(j)} \gamma_{i,r,0}}{\sum_{i=1}^n \sum_{r=1}^{d_j} x_{i,r}^{(j)}}, \quad \omega^{(j)}_{1} \gets 1 - \omega^{(j)}_{0}$
    \State $M_{jk,cr} \gets \frac{\sum_{i=1}^n x_{i,r}^{(j)} \gamma_{i,r,1} \pi_{i,r,j,c}}{\sum_{r'=1}^{d_j} \sum_{i=1}^n x_{i,r'}^{(j)} \gamma_{i,r',1} \pi_{i,r',j,c}}$
    \State $\eta_{j,r} \gets \frac{\sum_{i=1}^n x_{i,r}^{(j)} \gamma_{i,r,0}}{\sum_{i=1}^n \sum_{r'=1}^{d_j} x_{i,r'}^{(j)} \gamma_{i,r',0}}$
\Until{convergence}
\end{algorithmic}
\end{algorithm}

Since the $p(p-1)$ pairwise estimation problems are mutually independent, they can be executed in parallel. For each root candidate $j$, the marginal loss $\widehat{\mathcal{R}}_j^{\text{root}}$ is computed by simply setting $\widehat{\mathbf{x}}_i^{(j)} = \boldsymbol{\eta}_j$ and optimizing over $\boldsymbol{\eta}_j$ alone, which reduces to $\widehat{\boldsymbol{\eta}}_j = n^{-1}\sum_{i=1}^n \mathbf{x}_i^{(j)}$.

\subsection{Tree Search via Chu--Liu/Edmonds}\label{sec:treesearch}

After the pairwise estimation step, we obtain a $p \times p$ matrix of optimized edge losses $\widehat{\mathcal{R}}^*_{jk}$ for all ordered pairs $(j,k)$, as well as the root losses $\widehat{\mathcal{R}}^{\text{root}}_j$ for each node. For a given penalty $\alpha$, the sample-level penalized score~(\ref{eq:sample_score}) decomposes as a sum of independent per-node contributions. This decomposition maps the tree search exactly onto a minimum-weight spanning arborescence problem, which is solved in $O(p^2)$ time by the Chu--Liu/Edmonds algorithm \citep{chow1968approximating, edmonds1967optimum}.

A standard arborescence has a single root. To accommodate directed trees with multiple roots (i.e., forests), we augment the node set with a virtual root node~$0$ and add directed edges $0 \to j$ for every $j \in V$, with weight $\widehat{\mathcal{R}}^{\text{root}}_j + \alpha$. All other edge weights are set to $\widehat{\mathcal{R}}^*_{jk} + \alpha$ for $k \to j$. Running the Chu--Liu/Edmonds algorithm on this augmented graph yields an optimal arborescence rooted at node~$0$. Each edge $0 \to j$ selected in this arborescence indicates that node~$j$ is a root of the inferred tree/forest (i.e., it has no parent among the original nodes), while each edge $k \to j$ with $k \neq 0$ indicates that $k$ is the parent of $j$. Removing the virtual root and its incident edges recovers the optimal tree on $V$.

\subsection{Regularization Parameter Selection via Cross-Validation}\label{sec:cv}

The regularization parameter $\alpha$ controls the trade-off between fit and complexity: a smaller $\alpha$ favors denser trees with more edges. We select $\alpha$ by $K$-fold cross-validation. For each fold, the pairwise parameters and edge risks are estimated on the training set, the optimal tree is constructed for each candidate $\alpha$ via the Chu--Liu/Edmonds algorithm, and the predictive risk is evaluated on the held-out validation set. The predictive risk on the validation set for a given tree $\widehat{T}$ is $\sum_{j=1}^{p} \widehat{\mathcal{R}}_{j\pi_j}^{\text{val}}$, where $\widehat{\mathcal{R}}_{j\pi_j}^{\text{val}}$ is the risk of node $j$ given its parent $\pi_j$ in $\widehat{T}$ (or the root risk if $\pi_j = \varnothing$), evaluated using the training-set parameter estimates applied to the validation data. The value of $\alpha$ that minimizes the average validation risk across folds is selected, and the final tree is estimated using this $\widehat{\alpha}$ on the full dataset. The complete procedure is summarized in Algorithm~\ref{alg:cv}.


\begin{algorithm}
\caption{Structure Learning with Cross-Validation}\label{alg:cv}
\begin{algorithmic}[1]
\Require Data $\mathcal{D}$, Set of candidate penalties $\mathcal{A}$, Number of folds $K$
\Ensure Optimal tree $\widehat{T}$ and parameters $\widehat{\Theta}$
\State Partition $\mathcal{D}$ into $K$ disjoint folds $\{F_1, \dots, F_K\}$
\For{$t = 1$ to $K$}
    \State $\mathcal{D}_{\text{train}} \gets \mathcal{D} \setminus F_t$, \quad $\mathcal{D}_{\text{val}} \gets F_t$
    \State Estimate parameters and risks $\widehat{\mathcal{R}}^*_{jk}$, $\widehat{\mathcal{R}}^{\text{root}}_j$ on $\mathcal{D}_{\text{train}}$ using Algorithm~\ref{alg:em}
    \For{$\alpha \in \mathcal{A}$}
        \State Construct augmented graph with edge weights as described in \S\ref{sec:treesearch}
        \State $\widehat{T}_{\alpha, t} \gets \text{ChuLiuEdmonds}(V \cup \{0\},\, W)$
        \State Evaluate validation risk on $\mathcal{D}_{\text{val}}$ given $\widehat{T}_{\alpha, t}$ and training-set parameters
        \State Store validation risk in $S_{\alpha, t}$
    \EndFor
\EndFor
\State $\widehat{\alpha} \gets \arg\min_{\alpha \in \mathcal{A}} \frac{1}{K} \sum_{t=1}^K S_{\alpha, t}$
\State Re-estimate all risks on full dataset $\mathcal{D}$
\State $\widehat{T} \gets \text{ChuLiuEdmonds}(V \cup \{0\},\, W^*)$ where $W^*$ is based on the re-estimated risks and $\alpha = \widehat{\alpha}$.
\end{algorithmic}
\end{algorithm}

\section{Simulation Studies}\label{sec:sim}

We evaluate the proposed method in two synthetic settings that mirror the dimensionality and structural heterogeneity in our real-data applications. Simulation 1 uses $p=5$ compositional nodes and compares the proposed method with PC-stable and LiNGAM under repeated sampling. Simulation 2 increases graph size and complexity ($p=15$) and evaluates performance across chain, multi-root, and diverse branching structures, again benchmarking against PC-stable and LiNGAM.

Across both settings, the proposed method consistently attains the strongest overall recovery quality, with markedly lower false discovery rates while maintaining competitive true positive rates. These gains are most pronounced in higher-dimensional regimes. 
Full simulation details and complete results are reported in the Web Appendix C.

\section{Real Data Applications}\label{sec:realdata}

\subsection{Multi-Site Microbiome Data}\label{sec:realdata1}

We apply our method to the main motivating Multi-Omic Microbiome Study-Pregnancy Initiative (MOMS-PI) dataset \citep{fettweis2019vaginal}, accessed via the \texttt{HMP2Data} R package. The MOMS-PI study is a large-scale investigation of the maternal microbiome during pregnancy, with the primary objective of identifying microbial signatures associated with preterm birth and characterizing inter-site relationships of microbial communities within the maternal body. The study collected microbial samples from five body sites: Buccal mucosa, Rectum, Vagina, Feces, and Cervix of uterus.

We focus on the 96 subjects who provided complete samples at all five body sites at the baseline. For each subject, microbial profiles were obtained using rRNA-seq with operational taxonomic unit (OTU) reads summarized at the genus level. The raw data are extremely sparse, with an overall zero rate of 98.7\% across all OTUs. Following the common filtering protocol, we retain only OTUs meeting a prevalence threshold of $\ge 10\%$ and a mean abundance threshold of $\ge 5$. After filtering, the site-specific composition dimensions are: Buccal ($d=45$), Rectum ($d=55$), Vagina ($d=19$), Feces ($d=29$), and Cervix ($d=26$). As shown in Table~\ref{tab:zerorates}, the filtered compositions remain substantially zero-inflated, with zero rates ranging from 39.2\% (Vagina) to 67.4\% (Feces), underscoring the importance of a methodology that handles exact zeros gracefully.

\begin{table}[htbp]
\centering
\caption{Zero Rate Summary across 96 Subjects for the MOMS-PI dataset.}
\label{tab:zerorates}
\small
\begin{tabular}{@{}lcc@{}}
\toprule
\textbf{Body Site} & \textbf{Zero Rate Range} & \textbf{Overall Zero Rate}\\
\midrule
Buccal mucosa & 0 - 89.6\% & 46.1\% \\
Rectum & 0\% - 87.5\% & 47.4\% \\
Vagina & 0\% - 75.0\% & 39.2\% \\
Feces & 2.1\% - 88.5\% & 67.4\% \\
Cervix of uterus & 0\% - 88.5\% & 52.8\% \\
\bottomrule
\end{tabular}
\end{table}

Each body site yields a compositional vector for each subject, giving $p=5$ compositional nodes. We apply our method with leave-one-out cross-validation to select $\alpha$. The learned tree structure, shown in Figure~\ref{fig:momspi_structure}, identifies three disconnected components corresponding to distinct physiological systems: a digestive tract pathway (Rectum $\to$ Feces), a reproductive tract pathway (Vagina $\to$ Cervix), and the oral cavity (Buccal mucosa) as an isolated root.

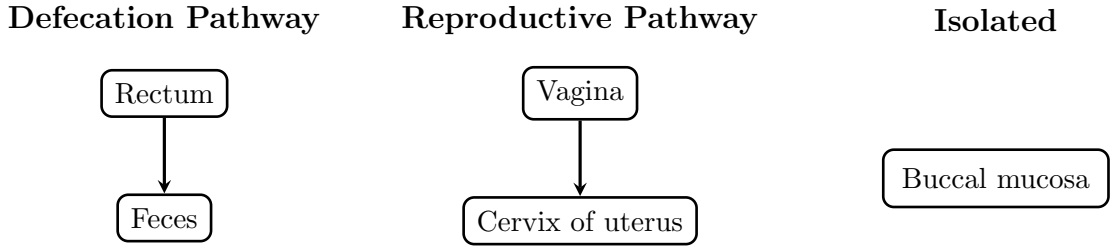
\begin{figure}[htbp]
    \centering
    \begin{tikzpicture}
        \node at (-5.5, 1.8) {\textbf{Defecation Pathway}};
        \node at (-5.5, 0) {
            \begin{forest}
            for tree={
                draw, rounded corners, font=\small, line width=1pt,
                edge={->, >=stealth, line width=1.2pt},
                parent anchor=south, child anchor=north,
                l sep=10mm, s sep=10mm, inner sep=5pt
            }
            [Rectum [Feces]]
            \end{forest}
        };

        \node at (0, 1.8) {\textbf{Reproductive Pathway}};
        \node at (0, 0) {
            \begin{forest}
            for tree={
                draw, rounded corners, font=\small,line width=1pt,
                edge={->, >=stealth, line width=1.2pt},
                parent anchor=south, child anchor=north,
                l sep=10mm, s sep=10mm, inner sep=5pt
            }
            [Vagina [Cervix of uterus]]
            \end{forest}
        };

        \node at (5.5, 1.8) {\textbf{Isolated}};
        \node[draw, rounded corners, font=\small, line width=1pt, inner sep=7pt] (buccal) at (5.5, -0.3) {Buccal mucosa};
    \end{tikzpicture}
    \caption{The learned tree structure for the MOMS-PI microbiome data. The model identifies three independent components corresponding to distinct physiological systems: the digestive tract (Rectum $\to$ Feces), the reproductive tract (Vagina $\to$ Cervix), and the oral cavity (Buccal mucosa).}
    \label{fig:momspi_structure}
\end{figure}

This structure aligns with well-established findings in microbial ecology. The directional link from Vagina to Cervix is supported by evidence of a microbiota continuum along the lower reproductive tract, with the cervical microbiome exhibiting greater similarity to the vaginal microbiome than to the uterine microbiome \citep{chen2017microbiota}.
The Rectum $\to$ Feces edge reflects the ecological continuity of the gastrointestinal tract, a pattern consistently reported in large-scale gut microbiome surveys \citep{qin2010human}. The isolation of Buccal mucosa is consistent with the strong site specificity of microbial communities documented by the Human Microbiome Project \citep{human2012structure}, which showed that oral microbial communities are compositionally distinct from both gastrointestinal and cervicovaginal communities. 

\begin{figure}[H]
\centering
\begin{tabular}{cc}
  \includegraphics[width=0.75\textwidth]{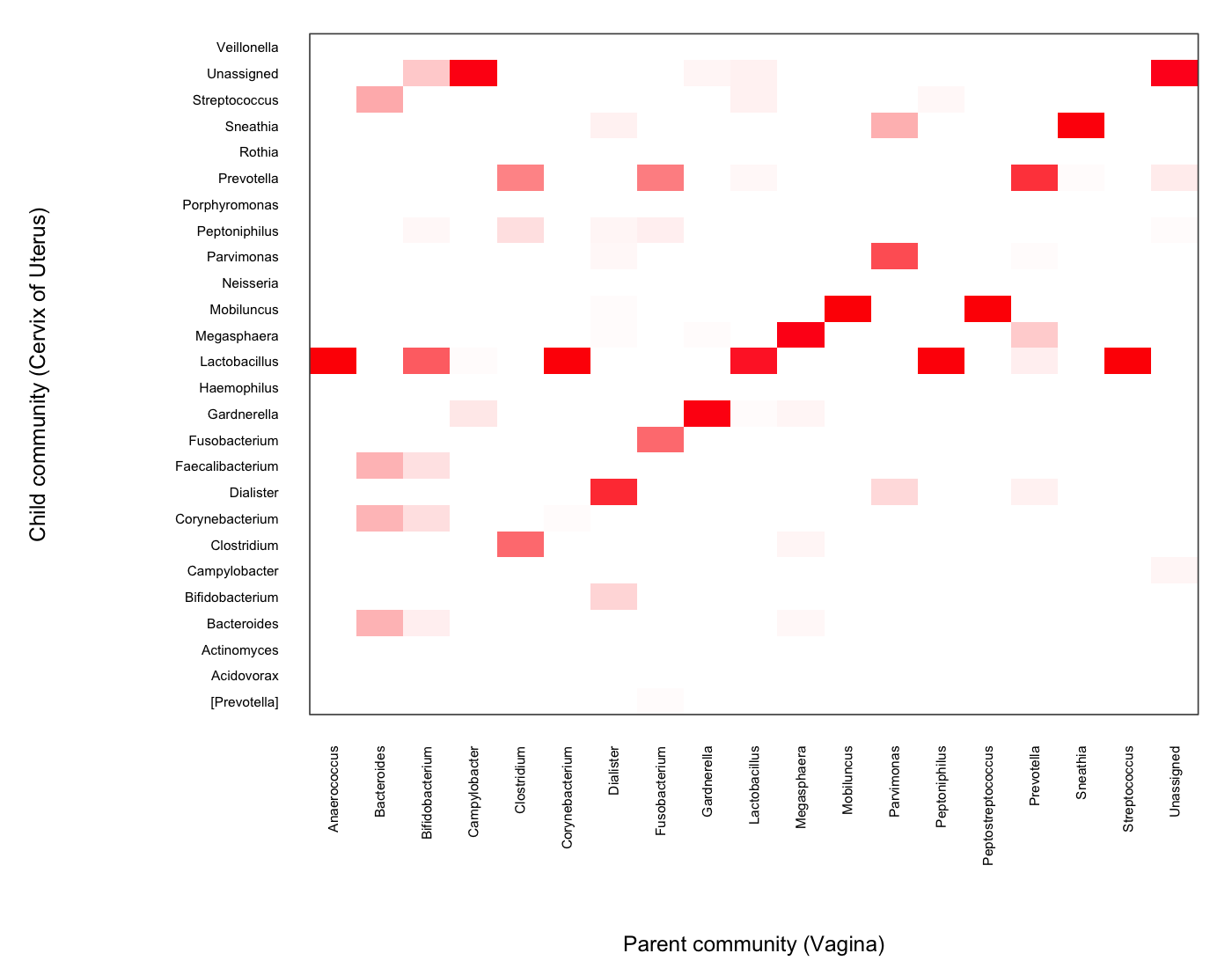} \\
  (a) Vagina \(\rightarrow\) cervix \\
  \includegraphics[width=0.75\textwidth]{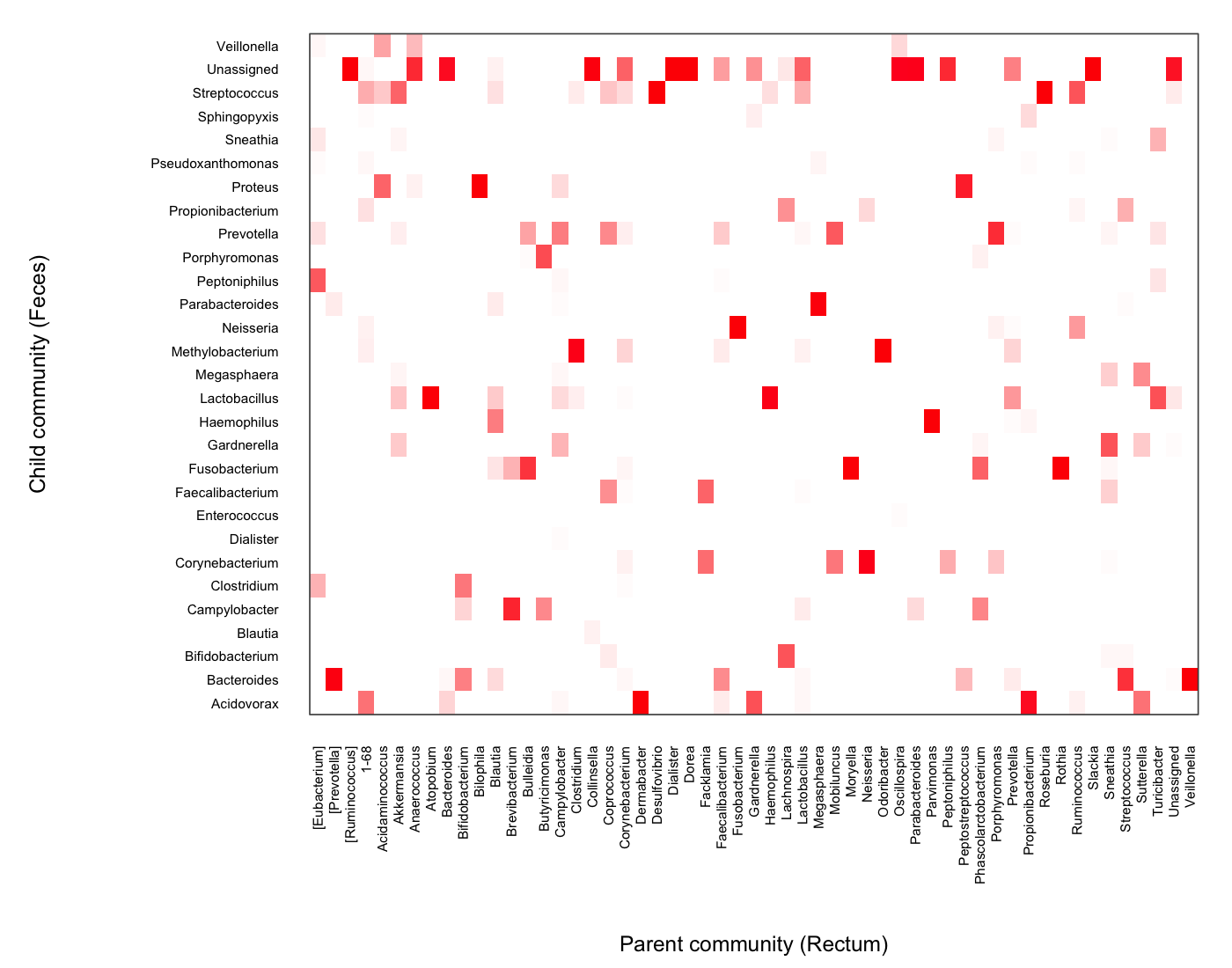} \\
  (b) Rectum \(\rightarrow\) feces
\end{tabular}
\caption{Estimated transition matrices $\mathbf{M}_{jk}$ for the two selected cross-site microbiome links: (a) vagina to cervix and (b) rectum to feces. The matrices display the transition weights between bacterial genera in the parent (x-axis) and child (y-axis) communities. Each column of $\mathbf{M}_{jk}$ sums to 1, with color intensity (white to red) indicating increasing weight.}
\label{fig:real-data-M}
\end{figure}

In Figure~\ref{fig:real-data-M}, we show the estimated transition matrices $\mathbf{M}_{jk}$ for the two selected edges. For the vagina to cervix link, $\widehat{\mathbf{M}}_{jk}$ exhibits high-weight same-genus mappings (e.g., \emph{Sneathia}\(\rightarrow\)\emph{Sneathia}, \emph{Gardnerella}\(\rightarrow\)\emph{Gardnerella}, \emph{Megasphaera}\(\rightarrow\)\emph{Megasphaera}, and \emph{Lactobacillus}\(\rightarrow\)\emph{Lactobacillus}), even though the model does not use any explicit genus-name information that encourages the  matching across sites. This pattern indicates that the model recovers coherent cross-site structure directly from the observed data. For the rectum to feces link, a high-weight cross-genus mapping is \emph{Veillonella}\(\rightarrow\)\emph{Bacteroides}. Both genera are common gut-associated taxa, and rectal and fecal microbial profiles are known to share substantial signal in comparative sampling studies \citep{rode2024fecal}. Overall, the recovered tree structure provides biological validation for the directional dependencies inferred by our model.

\subsection{Population-Scale Single-Cell Data}\label{sec:realdata2}

In the second application, we analyze the population-scale OneK1K single-cell cohort \citep{yazar2022single} and construct compositional nodes for immune cell subtypes via a preprocessing pipeline (dimension reduction, clustering, and GMM-based simplex representations). The resulting learned tree recovers biologically meaningful immune-development patterns, including links consistent with B-cell maturation and CD4$^+$ memory differentiation. The full data-processing pipeline and detailed results are provided in Web Appendix D.

\section{Discussion}\label{sec:con}

We have proposed a method for learning directed tree structures over zero-inflated compositional nodes, which is a setting not addressed by existing graphical models. The key ingredients are a KL-divergence-based score function that respects the simplex geometry and handles zero-inflated compositions, and a mixture-based conditional expectation model parameterized by a column-stochastic transition matrix. The former provides a principled measure of compositional discrepancy, whereas the latter ensures that predicted compositions remain on the simplex while enabling directional identifiability under a mild non-degeneracy condition. We established consistency of structure recovery and derived finite-sample guarantees, and demonstrated the method on both multi-site microbiome data, where the learned tree structure recovers known physiological pathways, and population-scale single-cell data, where the learned tree captures some established immune cell differentiation hierarchies.

Several directions merit further investigation. First, extending the framework to general DAGs would broaden applicability. This extension would require addressing both the computational challenge of searching over the super-exponential DAG space and the identifiability challenge of distinguishing among graphs within a Markov equivalence class (the asymmetry of the KL-based score may continue to provide leverage for orientation, as it does in the tree case). Second, incorporating uncertainty quantification for the learned edges, for example, via bootstrapping or a Bayesian formulation, would allow practitioners to assess confidence in individual edges rather than reporting only a point estimate of the tree. 




%
\bibliographystyle{apalike}
\bibliography{ref.bib}












\label{lastpage}

\end{document}